\documentclass[11pt]{article}

\usepackage{amsmath,amsfonts, amsthm, amssymb}
\usepackage{verbatim}

\usepackage[usenames,dvipsnames]{color}

\usepackage{hyperref}
\usepackage[margin=1in]{geometry}
\usepackage{palatino}
\renewcommand{\epsilon}{\varepsilon}
\renewcommand{\le}{\leqslant}
\renewcommand{\leq}{\leqslant}

\renewcommand{\geq}{\geqslant}

\newtheorem{thm}{Theorem}[section]\newtheorem*{rethm}{Theorem}
    \newtheorem{lemma}[thm]{Lemma}
    \newtheorem{cor}[thm]{Corollary}\newtheorem{prop}[thm]{Proposition}
    
\theoremstyle{definition}
    \newtheorem{defn}[thm]{Definition}
\theoremstyle{remark}\newtheorem*{rmk}{Remark}
\providecommand{\F}{\mathbb{F}}

\providecommand{\abs}[1]{\lvert#1\rvert}\providecommand{\gen}[1]{\langle#1\rangle}

\DeclareMathOperator{\spam}{span}\DeclareMathOperator{\poly}{poly}
\DeclareMathOperator{\proj}{proj}\DeclareMathOperator{\rk}{rank}
\DeclareMathOperator{\dist}{dist}

\title{Explicit rank-metric codes \\ list-decodable  with optimal redundancy\thanks{Research supported in part by NSF CCF-0963975.}}

\author{Venkatesan Guruswami\thanks{{\tt guruswami@cmu.edu}} \and Carol Wang\thanks{{\tt wangc@cs.cmu.edu}}}
\date{Computer Science Department \\ Carnegie Mellon University \\ Pittsburgh, PA 15213}

\begin{document}
\maketitle
\thispagestyle{empty}
\begin{abstract}

We construct an {\em explicit} family of {\em linear} rank-metric codes over any field $\F_h$ that enables efficient list decoding up to a fraction $\rho$ of errors in the rank metric with a rate of $1-\rho-\epsilon$, for any desired $\rho \in (0,1)$ and $\epsilon > 0$. Previously, a Monte Carlo construction of such codes was known, but this is in fact the first explicit construction of positive rate rank-metric codes for list decoding beyond the unique decoding radius.

\smallskip
Our codes are explicit subcodes of the well-known Gabidulin codes, which encode linearized polynomials of low degree via their values at a collection of linearly independent points. The subcode is picked by restricting the message polynomials to an $\F_h$-subspace that evades the structured subspaces over an extension field $\F_{h^t}$ that arise in the linear-algebraic list decoder for Gabidulin codes due to Guruswami and Xing (STOC'13). This subspace is obtained by combining subspace designs contructed by Guruswami and Kopparty (FOCS'13) with subspace-evasive varieties due to Dvir and Lovett (STOC'12).

\smallskip
We establish a similar result for subspace codes, which are a collection of subspaces, every pair of which have low-dimensional intersection, and which have received much attention recently in the context of network coding. We also give explicit subcodes of folded  Reed-Solomon (RS) codes with small folding order that are list-decodable (in the Hamming metric) with optimal redundancy, motivated by the fact that list decoding RS codes reduces to list decoding such folded RS codes. However, as we only list decode a {\em subcode} of these codes, the Johnson radius continues to be the best known error fraction for list decoding RS codes.

\end{abstract}

\newpage
\section{Introduction}

This paper considers the problem of constructing explicit list-decodable rank-metric codes. A {\em rank-metric code} is a collection of matrices $M\in \F_h^{n\times t}$ over a finite field $\F_h$ for fixed $n,t$. The rate of a rank-metric code is $\log_h\abs{\mathcal{C}}/(nt)$, and the distance measure between two codewords is the rank over $\F_h$ of their difference; that is, $\dist(M_1,M_2) = \rk_{\F_h}(M_1-M_2)$. We will be interested in {\em linear} rank-metric codes, where $\mathcal{C}$ is a subspace over $\F_h$. 

Rank-metric codes have found applications in network coding~\cite{SKK08} and public-key cryptography~\cite{GPT91,loidreau10}), among other areas. They can also be thought of as space-time codes over finite fields, and conversely can be used to construct space-time codes, eg. in \cite{LGB03,lu-kumar}. Unique decoding algorithms for rank-metric codes were shown in~\cite{FS} to be closely related to the so-called Low-rank Recovery problem, in which the task is to recover a matrix $M$ from few inner products $\gen{M,H}$. The authors of~\cite{FS} use their low-rank recovery techniques to construct rank-metric codes over any field, and show that they can be efficiently decoded.  
\smallskip

In this work, we will consider subcodes of {\em Gabidulin codes}, which are analogues of Reed-Solomon codes for the rank-metric. A Gabidulin code (denoted $\mathcal{C}_G(h;n,t,k)$) encodes $h$-linearized polynomials over $\F_{h^t}$ of $h$-degree less than $k$ by $\bigl( f(\alpha_1),\dotsc, f(\alpha_n)\bigr)^T$, where the $\alpha_i\in \F_{h^t}$ are linearly independent over $\F_h$, and $f(\alpha_j)$ is thought of as a column vector in $\F_h^t$ under a fixed basis of $\F_{h^t}$ over $\F_h$. This is a rank-metric code of rate $k/n$ and minimum distance $n-k+1$. 

We say that a rank-metric code $\mathcal{C}$ can be decoded from $e$ {\em rank errors} if any codeword $M\in\mathcal{C}$ can be recovered from $M+E$ whenever $E\in \F_h^{n\times t}$ has rank at most $e$. Gabidulin codes can be uniquely decoded from $(n-k)/2$ rank errors by adapting algorithms for Reed-Solomon decoding, as in \cite{Ga85, Ga91,roth91}, among others, but it is still open whether they can be {\em list-decoded} from a larger fraction of errors. We recall that in the list-decoding problem the decoder must output all codewords within the stipulated radius from the noisy codeword it is given as input. It is known that Gabidulin codes {\em cannot} be list-decoded  with a polynomial list size from an error fraction exceeding $1-\sqrt{R}$~\cite{faure,wachter-zeh}. However, as we show in this work, we can explicitly pick a good subcode of the Gabidulin code, with only a minor loss in rate, that enables efficient list-decoding all the way up to a fraction $(1-R)$ of errors.

\smallskip
The primary difficulty in previous work on list-decoding Gabidulin codes has been the fact that in contrast to Reed-Solomon codes, where the field size grows with the dimension of the code, for Gabidulin codes, the {\em dimension} of the ambient space grows with the dimension of the code. This forces us to work over fields whose size can be exponential in the code dimension.  

To address this, we show how to find linear list-decodable subcodes of certain Gabidulin codes by adapting the subspace designs of~\cite{GK} for use over large fields. The key observation, first made in~\cite{GX}, is that although applying a linear-algebraic list-decoder gives a subspace over a field which is too large, the subspace has additional structure which can then be ``evaded'' using {\em pseudorandom} subcodes, yielding a polynomial list size. 

We combine recent constructions of {\em subspace designs}~\cite{GK} and {\em subspace-evasive sets}~\cite{DL} in order to give an explicit construction of a subcode (in fact, subspace) of the Gabidulin code which has small intersection with the output of the linear-algebraic list-decoder of~\cite{GX}. In particular, we show (Theorem~\ref{thm:rank-metric}):
\begin{rethm}[Main] For every field $\F_h$, $\epsilon>0$ and integer $s>0$, there exists an explicit $\F_h$-linear subcode of the Gabidulin code $\mathcal{C}_G(h;n,t,k)$ with evaluation points $\alpha_1,\dotsc, \alpha_n$ spanning a subfield $\F_{h^n}$ that has (i) rate $(1-\epsilon)k/n$, and (ii) is list-decodable from $s(n-k)/(s+1)$ rank errors. The final list is contained in an $\F_h$-subspace of dimension $O(s^2/\epsilon^2)$. 
\end{rethm}

Note that the fraction of errors corrected approaches the information-theoretic limit of $(1-R)$ (where $R=k/n$ is the rate) as the parameter $s$ grows.
The authors of~\cite{GX} give a {\em Monte Carlo} construction of a subcode of the same Gabidulin code satisfying these guarantees, in fact with a better list size of $O(1/\epsilon)$. We give an {\em explicit} subcode, with a worse guarantee on the list size (which, however, is still bounded by a constant depending only on $\epsilon$).  

We also note that the above theorem gives  the {\em first explicit construction} of positive rate rank-metric codes even for list-decoding from a number of errors which is more than half the distance (and in particular for list decoding beyond a fraction $(1-R)/2$ of errors). Previous explicit codes only achieved polynomially small rate~\cite{GNW12}.

\medskip
Our techniques also imply analogous results for {\em subspace codes},
which can be thought of as a basis-independent form of rank-metric
codes. They were defined in~\cite{KK08} to address the problem of
non-coherent linear network coding in the presence of errors, and have
received much attention lately (\cite{ES09,MV-isit10,EV11}, etc). The
authors of~\cite{KK08} also define the K\"otter-Kschischang (KK)
codes, which, like Gabidulin codes, are linearized variants of
Reed-Solomon codes. List-decoding of a folded variant of the KK code
was considered in~\cite{GNW12} and~\cite{MV12}. However, both of these
papers could only guarantee a polynomial list size when the rate of
the code was polynomially small, and the question of constructing
constant rate list-decodable subspace codes remained open. Note
that~\cite{GX} was able, similarly to the case of rank-metric codes,
to give a Monte Carlo construction of a constant rate list-decodable
subcode.

In this work, we give the first explicit construction of high-rate subspace codes which are list-decodable past the unique decoding radius (stated in Theorem~\ref{thm:KK}). Our construction does not use folding, but instead takes subcodes of certain KK codes. 
\medskip

Additionally, we use our ideas to list-decode a subcode of the folded Reed-Solomon code where the folding parameter is of low order (see Corollary~\ref{thm:lo-frs} for a formal statement). List-decoding of the folded Reed-Solomon code up to list-decoding capacity where the folding parameter is primitive was first shown in~\cite{GR-FRS}. In~\cite{GW}, the authors use the linear-algebraic method to list-decode folded Reed-Solomon codes when the folding parameter has order at least the dimension of the code. 

\medskip
\noindent {\bf Paper organization.} In Section~\ref{sec:ntn}, we collect notation and definitions which will be used throughout the paper. In Section~\ref{sec:se}, we define and construct ``$(s,A,t)$-subspace designs,'' which is the new twist on the subspace designs of \cite{GK} that drives our results.
In Section~\ref{sec:rm}, we show how these subspace designs can be used to construct list-decodable rank-metric codes. In Section~\ref{sec:frs}, we give a list-decodable subcode of folded Reed-Solomon codes with low folding order. The construction of list-decodable subspace codes appears as Appendix~\ref{sec:KK}. 

We conclude in Section~\ref{sec:conc} with some open problems. 

\section{Notation and definitions}
\label{sec:ntn}
\vspace{-1ex}
Throughout the presentation of rank-metric codes, $\F_h$ is a finite field of constant size. $\F_q:=\F_{h^t}$ extends $\F_h$, and we will think of $\F_q$ as a vector space over $\F_h$ by fixing a basis. We will also have $n =  mt$, and the field $\F_{h^n}:= \F_{q^m}=\F_{h^{mt}}$ extending $\F_q$. 

In our final applications, $s$ will $\approx 1/\epsilon$, $m$ will be $\approx s/\epsilon$, where we will be list decoding up to error fraction $(1-\text{rate}-\epsilon)$, and $t$ will grow. 

We will be talking about subspaces over a field and its extension, so to avoid any confusion about the underlying field, we will usually refer to a subspace over a field $\F$ as an $\F$-subspace.

We recall some of the definitions of the pseudorandom objects concerning subspaces that we require. 

\begin{defn}[Strong subspace designs, \cite{GX}] A collection $S$ of $\F_q$-subspaces $H_1,\dotsc, H_M\subseteq\F_q^m$ is called a $(s,A)$ \emph{subspace design} if for every $\F_{q}$-linear space $W\subset \F_{q}^m$ of dimension $s$,
\[\sum_{i=1}^M \dim_{\F_q} (H_i\cap W)\leq A.\]
\end{defn}

\begin{defn}[Subspace-evasive sets, \cite{GW}] A subset $\mathcal{V}\subseteq \F_q^k$ is $(s,L)$ {\em subspace-evasive} if for every $\F_q$-subspace $S\subset \F_q^k$ of dimension $s$, $\abs{S\cap \mathcal{V}}\leq L$. 
\end{defn}

\section{Subspace designs}
\label{sec:se}
Throughout this section $q$ and $h$ will be prime powers with $q=h^t$. 
In what follows, we will think of subspaces $W\subseteq\F_q^m$ as $\F_h$-subspaces of $\F_h^{mt}$ via some fixed basis embedding. 

\begin{defn} A collection $S$ of $\F_h$-subspaces $H_1,\dotsc, H_M\subseteq\F_h^{tm}$ is called a $(s,A, t)$ \emph{$\F_h$-subspace design} if for every $\F_{h^t}$-linear space $W\subset \F_{h^t}^m$ of dimension $s$,
\[\sum_{i=1}^M \dim_{\F_h} (H_i\cap W)\leq A.\]
\end{defn}
\noindent
Note that in the above definition the dimension of the input $W$ is measured as a subspace over $\F_{h^t}$ whereas for the intersection, wh ich is an $\F_h$-subspace, the dimension is over $\F_h$.

\begin{rmk} When $t=1$, these are the (strong) subspace designs of \cite{GK}. We will be interested in settings where $t=\omega(1)$, so that considering $W$ as a subspace of dimension $st$ over $\F_h$ will generally not give strong enough bounds. 
\end{rmk}

\subsection{Existential bounds}

The following proposition shows that good subspace designs exist; indeed, a random collection of subspaces works with high probability. The $t=1$ case was established in \cite{GX12}.


\begin{prop} Let $\epsilon>0$. Let $S$ consist of $M=h^{\epsilon tm/8}$ $\F_h$-subspaces of codimension $\epsilon tm$ in $\F_h^{mt}$, chosen independently at random. Then for any $s<m\epsilon/2$, with probability at least $1-q^{-ms}$, $S$ is a $(s, 8s/\epsilon,t)$ $\F_h$-subspace design. (Here $q=h^t$.)
\end{prop}
\begin{proof} Set $\ell=8s/\epsilon$, and let $S=\{H_1,\dotsc, H_M\}$. For a fixed $\F_{h^t}$ subspace $W$ of dimension $s$ and any $j$, the probability that $\dim_{\F_h}(W\cap H_j)\geq a$ at most $q^{sa}\cdot q^{-\epsilon ma}\leq q^{-\epsilon ma/2}$, by assumption on $s$. 

Since the $H_i$ are independent, for a fixed tuple $(a_1,\dotsc, a_M)$ of nonnegative integers summing to $\ell=8s/\epsilon$, the probability that $\dim(W\cap H_i)\geq a_j$ for each $j$ is at most $q^{-\epsilon m\ell/2}=q^{-4ms}$. Union bounding over the at most $q^{ms}$ choices of $W$ and $\binom{\ell+M}{\ell}\leq M^{2\ell}$ choices of $(a_1,\dotsc, a_M)$, the probability $S$ is \emph{not} a $(s,8s/\epsilon,t)$  $\F_h$-subspace design is at most 
\[q^{ms}M^{2\ell}\cdot q^{-4ms}=q^{ms}\cdot q^{2ms}\cdot q^{-4ms}\leq q^{-ms} \ . \qedhere \]
\end{proof}

\subsection{Constructive bounds}

In this section, we show how to construct an explicit large $\bigl(s, 2(m-1)s/\epsilon,t\bigr)$ $\F_h$-subspace design  consisting of $\F_h$-subspaces of $\F_h^{tm}$ of co-dimension $2\epsilon tm$. 

The idea, which is natural in hindsight, is to first use a subspace design over $\F_{h^t}$ to ensure that the intersection with any $\F_{h^t}$-subspace of dimension $s$ has low dimension over $\F_{h^t}$, and then to use a subspace-evasive set to reduce the dimension further over $\F_h$. The final construction appears as Theorem~\ref{thm-design}. 

\subsubsection{Explicit subspace-evasive sets}

We first describe the construction of explicit subspace-evasive sets which we will be using. 
\medskip

Let $q>h^{m-1}$, and let $\gamma_1,\dotsc, \gamma_m$ be distinct elements of $(\F_q)^*$. Let $A$ be the $s\times m$ matrix with $A_{ij} = \gamma_j^i$. Then Dvir and Lovett~\cite{DL} showed the following: 

\begin{thm}
\label{thm:DL-statement}
 Let $1 \leq s\leq m$. Let $d_1>d_2>\dotsb > d_m\geq 1$ be integers. Define $f_1,\dotsc, f_s\in \F_q[X_1,\dotsc, X_m]$ as follows:
\begin{equation}
\label{eqn:dl}
f_i(x_1,\dotsc, x_m) = \sum_{j=1}^m A_{ij} x_j^{d_j} \ . 
\end{equation}
Then:
\begin{itemize}
\item The variety $\mathbf{V}=\{x\in \overline{F}_q^m\mid f_1(x)=\dotsb= f_s(x)=0\}$ satisfies $\abs{\mathbf{V}\cap H}\leq (d_1)^s$ for all $s$-dimensional affine subspaces $H\subset \overline{\F}_q^m$. 
\item 
If at least $s$ of the degrees $d_i$ are relatively prime to $q-1$, then $\abs{\mathbf{V}\cap \F_q^m}= q^{m-s}$. 
\end{itemize}
\noindent
Additionally, the product set $(\mathbf{V}\cap \F_q^m)^{n/m}\subseteq \F^n$ is $(k, (d_1)^k)$-subspace evasive for all $k\leq s$. 
\end{thm}

The below statement follows immediately from Theorem~\ref{thm:DL-statement} and the fact that when the $d_j$'s are powers of $h$, the polynomials $f_i$ defined in \eqref{eqn:dl} are $\F_h$-linear functions on $\F_q^m$.
\begin{cor} \label{lem:se} Setting $d_1=h^{m-1},d_2=h^{m-2},\dotsc, d_m=1$, we obtain an explicit $\F_h$-linear set $S$ of size $q^{(m-s)n/m}$ over $\F_q^n$ which is $(k, h^{(m-1)k})$ subspace-evasive for all $1\leq k\leq s$. 
\end{cor}
\begin{rmk} One can improve on the degree bounds and therefore the final intersection size via a standard subspace-evasive set without the $\F_h$-linearity requirement. For example, \cite{DL}~gives a construction of a (non-linear) $\bigl(s, (s/\epsilon)^s\bigr)$ subspace-evasive set over $\F^n$ of size $\abs{\F}^{(1-\epsilon) n}$. 

However, especially in applications for rank-metric codes, linearity is a property which is desirable and often necessary. 
\end{rmk}

\subsubsection{Combining with subspace designs}

The following theorem shows how to achieve our initial goal of ensuring small intersection dimension over the larger field $\F_{h^t}$. 

\begin{thm}[\cite{GK}] \label{thm:gk} For $\epsilon \in (0,1)$, positive integers $s,m$ with $s\leq \epsilon m/4$, and $q>m$, there is an explicit collection of $M=q^{\Omega(\epsilon m/s)}$ subspaces in $\F_q^m$, each of codimension at most $\epsilon m$, which form a $(s, 2s/\epsilon,1)$ $\F_q$-subspace design. 
\end{thm}

Combined with Corollary~\ref{lem:se}, we now have a construction of a $(s, 2(m-1)s/\epsilon,t)$ $\F_h$-subspace design, summarized in the following statement.

\begin{thm} \label{thm-design} For integers $s\leq \epsilon m/4$ and $q>m$, there exists an explicit set of  $q^{\Omega(\epsilon m/s)}$ $\F_h$-subspaces in $\F_{h}^{tm}$ of co-dimension at most $2\epsilon tm$ forming a $(s, 2(m-1)s/\epsilon,t)$ $\F_h$-subspace design. 
\end{thm}

\begin{proof} Let $V_1,\dotsc, V_M \subseteq \F_q^m$ be the elements of the $(s, 2s/\epsilon,1)$ $\F-q$-subspace design of Theorem~\ref{thm:gk}. For each $i$, define $H_i = V_i\cap S$, where $S\subseteq\F_q^m$ is the $(s, h^{(m-1)s})$ subspace-evasive set of Corollary~\ref{lem:se}. As $S$ and the $V_i$'s are $\F_h$-linear subspaces, $H_i$ is as well. We claim that the $H_i$'s form the desired $\F_h$-subspace design.

For each $i$, $V_i$ has co-dimension $\epsilon tm$, and $S$ has co-dimension $ts\leq \epsilon tm/4$, so the co-dimension of $H_i$ is at most $2\epsilon tm$. 

Now let $W$ be an $\F_q$-subspace of dimension $s$. By the $\F_q$-subspace design property of the $V_i$'s we have
\begin{equation}
\label{eq:gk-guarantee}
\sum_{i=1}^M \dim_{\F_q}(V_i \cap W) \le 2s/\epsilon \ .
\end{equation}
For each $i$, we also have that $\dim_{\F_q}(W\cap V_i) = s_i \leq s$, so by the subsace evasive property of $S$ from Corollary~\ref{lem:se}, $W \cap H_i = (W \cap V_i) \cap S$ has at most $h^{(m-1)s_i}$ elements. As $W \cap H_i$ is $\F_h$-linear, we have 
\begin{equation}
\label{eq:dl-guarantee}
\dim_{\F_h}(W \cap H_i) \le (m-1) \dim_{\F_q}(W \cap V_i) \ .
\end{equation}
Combining \eqref{eq:gk-guarantee} and \eqref{eq:dl-guarantee} we have 
\[\sum_i \dim_{\F_h}(W\cap H_i) \leq  \sum_i (m-1)\dim_{\F_q}(W\cap V_i)\leq (m-1)\cdot 2s/\epsilon \ .  \qedhere \]
\end{proof}

The motivation for constructing the above subspace design is that they yield a subspace that has small intersection with so-called periodic subspaces arising in certain linear-algebraic list decoding algorithms. We recall the definition from~\cite{GX}. Below, for a string $\mathbf{x}=(x_1,x_2,\dots,x_\ell)$, we denote by $\proj_{[a,b]}(\mathbf{x})$ the substring $(x_a,x_{a+1},\dots,x_b)$.

\begin{defn}[Periodic subspaces] 
\label{defn:periodic} For positive integers $s,m,k$ and $\kappa:=mk$, an affine subspace $H\subset \F_q^\kappa$ is said to be $(s,m,k)$-\textbf{periodic} if there exists a subspace $W\subseteq\F_q^m$ of dimension at most $s$ such that for every $j=1,2,\dotsc, k$, and every prefix $\mathbf{a}\in \F_q^{(j-1)m}$, the projected affine subspace of $\F_q^m$ defined by
\[\{\proj_{[(j-1)m+1,jm]}(\mathbf{x})\mid \mathbf{x}\in H~\text{and }\proj_{[1,(j-1)m]}(\mathbf{x})=\mathbf{a}\}\]
is contained in an affine subspace of $\F_q^m$ given by $W+\mathbf{v}_{\mathbf{a}}$ for some vector $\mathbf{v}_{\mathbf{a}}\in\F^m$ dependent on $\mathbf{a}$. 
\end{defn}

\begin{prop} \label{lem:periodic}
Let $H$ be a $(s, m, k)$-periodic affine suspace of $\F_q^{mk}$, and $H_1,H_2,\dots,H_k \subseteq \F_h^{mt}$ be distinct subspaces from a $(s, A,t)$ $\F_h$-subspace design. Then $H\cap (H_1\times\dotsb\times H_k)$ is an affine subspace over $\F_h$ of dimension at most $A$. 
\end{prop}
\begin{proof} It is clear that $H\cap (H_1\times\dotsb\times H_k)$ is an affine subspace over $\F_h$. Let $W$ be the subspace associated to $H$ as in Definition~\ref{defn:periodic}. We will show by induction that $\abs{\proj_{[1, im]}(H)\cap (H_1\times\dotsb\times H_i)}\leq h^{\sum_{j=1}^i \dim_{\F_h} (W\cap H_j)}$. 

In the base case, since $H_1$ is a subspace, $\proj_{[1,m]}(H)\cap H_1=(W+v_{\mathbf{0}})\cap H_1$ is an affine subspace whose underlying subspace lies in $W\cap H_1$. In particular, its size is at most $h^{\dim(W\cap H_1)}$. 

Continuing, fix an element $\mathbf{a}\in \proj_{[1, im]}(H)\cap (H_1\times\dotsb\times H_i)$. Because $H$ is periodic and $H_{i+1}$ is linear, the possible extensions of $\mathbf{a}$ in $\proj_{[im+1, (i+1)m]}(H)\cap H_{i+1}$ are given by a coset of $W\cap H_{i+1}$. Thus, there are at most $h^{\dim(W\cap H_{i+1})}$ such extensions. Since by induction there were $h^{\sum_{j=1}^i \dim_{\F_h} (W\cap H_j)}$ possibilities for the prefix $\mathbf{a}$, the result follows. 
\smallskip

In particular, $H\cap (H_1\times\dotsb\times H_k)$ has dimension over $\F_h$ which is at most $\sum_{i=1}^k \dim(W\cap H_i)\leq A$, by the subspace design property. 
\end{proof}

\section{Explicit list-decodable rank-metric codes}
\label{sec:rm}

In this section, we show how to use the subspace designs of Theorem~\ref{thm-design} in order to get explicit list-decodable rank-metric codes of optimal rate for any desired error correction radius.
\medskip

We first review rank-metric codes, and in particular the Gabidulin code~\cite{Ga85}, which is the starting point of our construction. 

Let $h$ be a prime power, and let $\mathbb{M}_{n\times t}(\F_h)$ be the set of $n\times t$ matrices over $\F_h$. The {\em rank distance} between $A,B\in\mathbb{M}_{n\times t}(\F_h)$ is $d(A,B)=\rk(A-B)$. A rank-metric code $\mathcal{C}$ is a subset of $\mathbb{M}_{n\times t}(\F_h)$, with rate and distance given by 
\[R(\mathcal{C}) = \frac{\log_h\abs{\mathcal{C}}}{nt}\quad \text{ and }\quad d(\mathcal{C}) = \min_{A\neq B\in\mathcal{C}}\{d(A,B)\}.\]

The {\em Gabidulin code} encodes $h$-linearized polynomials of by their evaluations at linearly independent points. Recall that an $h$-linearized polynomial $f$ over $\F_{h^t}$ is a polynomial of the form $\sum_{i=0}^\ell a_i X^{h^i}$, with $a_i\in\F_{h^t}$. If $a_\ell\neq 0$, then $\ell$ is called the $h$-{\em degree} of $f$. We write $\mathcal{L}_h(t)$ for the set of $h$-linearized polynomials over $\F_{h^t}$. 

Let $0<k\leq n\leq t$ be integers, and choose $\alpha_1,\dotsc, \alpha_n\in\F_{h^t}$ to be linearly independent over $\F_h$. For every $h$-linearized polynomial $f\in \F_{h^t}[X]$ of $h$-degree at most $k-1$, we can encode $f$ by the column vector $M_f=\bigl( f(\alpha_1),\dotsc, f(\alpha_n)\bigr)^T$ over $\F_{h^t}$. By fixing a basis of $\F_{h^t}$ over $\F_h$, we can also think of $M_f$ as an $n\times t$ matrix over $\F_h$. This yields the Gabidulin code
\[\mathcal{C}_G(h;n,t,k) := \{M_f\in \mathbb{M}_{n\times t}(\F_h)\mid f\in\mathcal{L}_h(t),\,h\text{-degree}(f)\leq k-1\}.\]
If a rank-metric codeword $X$ is transmitted, and a matrix $Y$ is received, we say that $\rk(Y-X)$ {\em rank errors} have occurred. 

Suppose that $t=nm$ for some integer $m$, so that $\F_{h^t}$ has a subfield $\F_{h^n}=: \F_q$. In the case when the evaluation points $\alpha_1,\dotsc, \alpha_n$ of the Gabidulin code span $\F_{h^n}$, Guruswami and Xing~\cite{GX} show the following:

\begin{thm}[\cite{GX}] Let $f\in \F_{h^t}[X]$ be an $h$-linearized polynomial with $h$-degree at most $k-1$. Suppose that a codeword $M_f=\bigl( f(\alpha_1),\dotsc, f(\alpha_n)\bigr)^T$ is transmitted and $Y=(y_1,\dotsc, y_n)^T$ is received with at most $e$ rank errors. If $e\leq \frac{s(n-k)}{s+1}$, then there is an algorithm running in time $\poly(n, m, \log q)$ outputting a $(s-1, m,k)$-periodic subspace containing all candidate messages $f$. 
\end{thm}

By Proposition~\ref{lem:periodic}, by restricting the message polynomials $f=\sum_i f_i X^{q^i}$ to have coefficients $f_i \in H_{i+1}$ for $0 \le i < k$, where $H_1,H_2,\dots,H_k$ are distinct elements of the subspace design in Theorem~\ref{thm-design}, the final list of candidate messages will have dimension at most $2(m-1)s/\epsilon$ over $\F_h$, or size at most $h^{2(m-1)s/\epsilon}$. As one can take $m = O(s/\epsilon)$ for the necessary subspace design guaranteed by Theorem~\ref{thm-design}, we can conclude the following theorem, which is our main result.
\begin{thm} \label{thm:rank-metric} For every $\epsilon>0$ and integer $s>0$, there exists an explicit $\F_h$-linear subcode of the Gabidulin code $\mathcal{C}_G(h;n,t,k)$ with evaluation points spanning $\F_{h^n}$ of rate $(1-2\epsilon)k/n$ 
which is list-decodable from $\frac{s}{s+1}\cdot (n-k)$ rank errors. The final list is contained in an $\F_h$-subspace of dimension at most $O(s^2/\epsilon^2)$. 
\end{thm}

\section{Application to low-order folding of Reed-Solomon codes}
\label{sec:frs}

In this section, we show how the idea of only evading subspaces over an extension field can be used to give an algorithm for list-decoding (subcodes of) folded Reed-Solomon codes in the case when the folding parameter has low ($O(1)$) order. 

As in the case of KK codes, our decoding algorithm follows the framework of interpolating a linear polynomial and then solving a linear system for candidate polynomials. 
\medskip
Fix $\gamma$ generating $\F_q^*$. Let $N=\frac{q-1}{\ell}$, and let $\zeta=\gamma^{N}$, which has order $\ell$ in $\F_q$. 
Then the {\bf low-order folded Reed-Solomon code} encodes a polynomial $f$ of degree $<k$ by 
\[f\mapsto\begin{bmatrix}
f(1) & f(\gamma) & \dotsb & f(\gamma^{N-1})\\
f(\zeta) & f(\zeta \gamma) & \hdots & f(\zeta \gamma^{N-1})\\
\vdots & \vdots & \ddots & \vdots \\
f(\zeta^{\ell-1})& f(\zeta^{\ell-1} \gamma) & \hdots & f(\zeta^{\ell-1} \gamma^{N-1})
\end{bmatrix}.\]

Similarly to folded Reed-Solomon codes, this is a code of rate $\frac{k}{\ell N}$ and distance $N - (k-1)/\ell$. 

\subsection{Interpolation}

Given a received word 
\[\begin{pmatrix} 
y_{00} & y_{01} & \hdots & y_{0(N-1)}\\ 
y_{10} & y_{11} & \hdots & y_{1 (N-1)}\\
\vdots & \vdots & \ddots & \vdots\\
y_{(\ell-1) 0} & y_{(\ell-1) 1} & \hdots & y_{(\ell-1) (N-1)}\\
\end{pmatrix},\]
we would like to interpolate a (nonzero) polynomial
\[Q(X, Y_1,\dotsc, Y_s)= A_0(X) + A_1(X) Y_1 + \dotsb + A_s(X) Y_s\]
such that 
\begin{equation}\label{interp}
Q\bigl(\gamma^{iN + j}, y_{ij}, y_{(i+1)j},\dotsc, y_{(i+s-1)j}\bigr)=0\qquad i\in \{0,\dotsc, \ell-1\},\,j\in \{0,\dotsc, N-1\},
\end{equation}
where all indices are taken modulo $\ell$. 
\medskip

We will require $\deg(A_0) \leq  D+k-1$, and $\deg(A_i)\leq D$ for $i>0$. 

\begin{lemma} Let 
\[D=\left\lfloor\frac{\ell N - k+1}{s+1}\right\rfloor.\] 
Then a nonzero polynomial $Q$ satisfying~\eqref{interp} exists (and can be found by solving a linear system). 
\end{lemma}
\begin{proof} The number of interpolation conditions is $\ell N$. The quantity $(D+1)(s+1)  + k-1>\ell N$ is the number of degrees of freedom for the interpolation, and the conditions are homogeneous, so a nonzero solution exists. 
\end{proof}

\begin{lemma} \label{lem:interp-sat}
If the number of agreements $t$ is greater than $\frac{D + k-1}{\ell}$,
then 
\begin{equation}
\label{eq:interp-sat}
Q\bigl(X, f(X), f(\zeta X),\dotsc, f(\zeta^{s-1}X)\bigr)=0. 
\end{equation}
\end{lemma}

\begin{proof} $Q\bigl(X, f(x),\dotsc, f(\zeta^{s-1}X)\bigr)$ is a univariate polynomial of degree $D+k-1$, and each correct column $j$ yields $\ell$ distinct roots $\gamma^{iN + j}$ for $i\in \{0,\dotsc, \ell-1\}$. Thus if $t\ell>\deg D+k-1\geq \deg Q$, $Q$ is the zero polynomial. 
\end{proof}

\noindent
For our choice of $D$, the requirement on $t$ in Lemma~\ref{lem:interp-sat} is met if $t$ satisfies
\begin{equation}\label{eq:agree}
\frac{t}{N}\geq \frac{1}{s+1} + \frac{s}{s+1}R.
\end{equation}

\begin{rmk} In ordinary folded Reed-Solomon codes, where the folding parameter is primitive of order $q-1$, the agreement fraction required to satisfy~\eqref{eq:interp-sat} is 
\[\frac{t}{N}\geq \frac{1}{s+1} + \frac{s}{s+1} \frac{\ell R}{\ell - s+1},\]
which is higher than~\eqref{eq:agree}. In our case, because $\zeta$ has low order, we are able to use interpolation conditions that ``wrap around,'' allowing us to impose $\ell$ conditions per coordinate rather than $\ell -s+1$. Therefore we can satisfy Equation~\eqref{eq:interp-sat} with lower agreement. On the other hand, it is known how to list-decode folded Reed-Solomon codes themselves, whereas we are only able to list-decode a subcode. 
\end{rmk}

\subsection{Decoding}

In this section, we describe how to solve the system 
\[Q\bigl(X, f(X), f(\zeta X),\dotsc, f(\zeta^{s-1}X)\bigr)=0\tag{\ref{eq:interp-sat}}\]
for candidate polynomials $f$. 

\begin{prop} \label{lem:decoding}
Given an irreducible polynomial $R(X)\in\F_q[X]$ such that 
\begin{itemize}
\itemsep=0ex
\item $\deg R\geq k$, and
\item for some $a$, $\zeta X \equiv X^{q^a}\pmod{R}$. 
\end{itemize}
Then the set of $f$ of degree $<k$ satisfying \eqref{eq:interp-sat} is an $\F_{q^a}$-affine subspace of dimension at most $s-1$. 
\end{prop}

\begin{proof} The condition \eqref{eq:interp-sat} says 
\[0=A_0(X) + A_1(X) f(X) + A_2(X) f(\zeta X) + \dotsb + A_s(X) f(\zeta^{s-1}X).\]

Then we have 
\[A_0(X) + A_1(X) f(X) + A_2(X) f(X)^{q^a} + \dotsb + A_s(X) f(X)^{q^{(s-1)a}}\equiv 0\pmod{R}.\]

By dividing out the highest power of $R$ which divides every $A_i$, Equation~\eqref{eq:interp-sat} is still satisfied and we may assume that this equation is nonzero mod $R$. 

In particular, this equation has at most $q^{(s-1)a}$ solutions for $f$ mod $R$. When $\deg f<k\leq \deg R$, $f$ is uniquely determined by its residue mod $R$ and there are at most $q^{(s-1)a}$ solutions for $f$. 

The fact that the solution space is $\F_{q^a}$-affine follows from the fact that the terms in which $f(X)$ appears all have degree $q^{ai}$ for some $i$. 
\end{proof}

Because the output space is a subspace (over the large field
$\F_{q^a}$), by picking the message polynomials $f$ to come from a
subspace-evasive set, we can reduce the list size bound. More
specifically, if $\ell$ is at least $s/\epsilon$, \cite{DL} gives a
construction of a $(s, (s/\epsilon)^{s})$ subspace-evasive set $S$
over $(\F_{q^a})^{k/a}$ of size $q^{(1-\epsilon) k}$. By precoding the
messages to come from this set $S$, we are able to both encode and
compute the intersection of the code with the output subspace of
Proposition~\ref{lem:decoding} in polynomial time.

Setting $s=O(1/\epsilon)$ and $\ell=O(s/\epsilon)$, we obtain the following. 
\begin{cor} \label{thm:lo-frs}
For every $\epsilon>0$ and $R\in(0,1)$, there is an explicit rate $R$ subcode of a low-order folded Reed-Solomon code which is list-decodable from a $1-R-\epsilon$ fraction of errors with list size $(1/\epsilon)^{O(1/\epsilon)}$, given an irreducible polynomial satisfying the conditions of Proposition~\ref{lem:decoding}. 
\end{cor}

\begin{rmk}
By using Corollary~\ref{lem:se} instead of the results of \cite{DL},
we can give a similar guarantee which yields a {\em linear} subcode,
but with a larger list size guarantee of $q^{\poly(1/\epsilon)}$.

The techniques of~\cite{GX} using subspace designs could also be
applied directly to the case of low-order folding, with a resulting
list size of $n^{\poly(1/\epsilon)}$. We are able to get an
improvement using the observation that the space of candidates is
actually a low-dimensional subspace over a much larger field.
\end{rmk}

\subsection{Constructing high-degree irreducibles}
\label{sec:irred}
The decoding algorithm of the previous section relied on working modulo a high-degree irreducible factor of $X^{q^a}-\zeta X$. In what follows, we consider the problem of finding such a factor efficiently.
\begin{prop} \label{lem:deg-cond}
For $\zeta\in \F_q$ of order $\ell$, the irreducible factors over $\F_q[X]$ of
\[X^{q^a-1}-\zeta \]
have degree dividing $a\ell$. In particular, all roots of $X^{q^a-1}-\zeta$ lie in $\F_{q^{a\ell}}$. 
\end{prop}
\begin{proof} As $X^{(q^a-1)\ell} \equiv 1\pmod{X^{q^a-1}-\zeta}$, it is enough to see that $(q^a-1)\ell$ divides $q^{a\ell}-1$. This implies that $X^{q^a-1}-\zeta$, and thus all of its irreducible factors, divides $X^{q^{a\ell}}-X$. 

As $\ell\mid q-1$, we have
\[\frac{q^{a\ell}-1}{q^a-1}=q^{a(\ell-1)} + q^{a(\ell-2)} + \dotsb +q^a +1\equiv 0\pmod{\ell} \ . \qedhere \]
\end{proof}

\begin{cor} \label{lem:a}
If $a$ and $\ell$ with $a>2\ell$ are distinct primes, at least half of the roots of $X^{q^a-1}-\zeta$ have irreducible polynomials of degree $a\ell$. 
\end{cor}
\begin{proof} By Proposition~\ref{lem:deg-cond}, all irreducible factors of $X^{q^a-1}-\zeta$ have degrees in the set $\{1,a, \ell, a\ell\}$. No irreducible factor has degree $1$ or $a$, because any irreducible of degree $1$ or $a$ divides $X^{q^a-1}-1$ and therefore does not divide $X^{q^a-1}-\zeta$ for $\zeta\neq 1$. 

Because $X^{q^a-1}-\zeta$ has no repeated factors, it has at most $q^\ell$ roots which lie in $\F_{q^\ell}$ (and hence have irreducible polynomials of degree $\ell$. 

Thus, under the assumptions on $a$ and $\ell$, $X^{q^a-1}-\zeta$ has at least $(q^a-q^\ell-1)\geq q^\ell$ roots of degree $a\ell$. Thus at least half of of $X^{q^a-1}-\zeta$'s roots have irreducible polynomials of degree $a\ell$. 
\end{proof}

In particular, by choosing $a$ to be a prime in the range $[k/\ell, 2k/\ell]$, we have $k\leq a\ell\leq 2k$, so that an irreducible factor of $X^{q^a-1}-\zeta$ will satisfy the conditions of Proposition~\ref{lem:decoding}. The next section will show that we cannot hope to improve much on the value of $a$.
\medskip

Given a value for $a$ for which $X^{q^a-1}-\zeta$ has many degree $a\ell$ factors, the problem remains to compute one. In what follows, we describe one randomized approach. 
\medskip

Recall that $a$ and $\ell$ are primes, and that we are trying to find a degree $a\ell$ factor of $X^{q^a-1}-\zeta$. The idea is to sample a root of $X^{(q^a-1)\ell}-1$. Consider the following procedure: 
\begin{enumerate}
\itemsep=0ex
\item Sample $\beta\in (\F_{q^a})^*$ uniformly at random. 
\item Compute the roots $\rho_1,\dotsc, \rho_\ell$ of $X^\ell -\beta$, which lie in $\F_{q^{a\ell}}$ by Proposition~\ref{lem:deg-cond}. This can be done in time $\tilde{O}(n^2\log (q^a) \log^{-1}\epsilon)$ with failure probability $\epsilon$ using a variant of Berlekamp's algorithm (see, for example, \cite{K92}). 
\item Compute $\rho_i^{q^a-1}$ for each $i$ and output the minimal polynomial of $\rho_i$ over $\F_q$ if $\rho_i^{q^a-1}=\zeta$. 
\end{enumerate}

First note that steps 1--2 sample each root of $X^{(q^a-1)\ell}-1$ uniformly. Each $\rho_i$ computed in step 2 satisfies $\rho_i^\ell\in(\F_{q^a})^*$, so $\rho_i$ is a root of $X^{(q^a-1)\ell}-1$. Conversely, each nonzero $\beta$ yields $\ell$ distinct roots of $X^\ell-\beta$, which are distinct for distinct $\beta$, yielding $(q^a-1)\ell$ roots. 

Therefore, with probability $1/\ell$, we will find a root $\rho$ of $X^{q^a-1}-\zeta$. By Corollary~\ref{lem:a}, $\rho$'s minimal polynomial has degree $a\ell$ with probability at least $1/2$. 

We can thus conclude that, with probability at least $\frac{1}{2\ell} - \epsilon$, we find an irreducible factor of $X^{q^a-1}-\zeta$ of degree $a\ell$.

\subsection{Relationship to Reed-Solomon list-decoding} 
\label{sec:lb}

The original motivation for studying low-order folding was the following reduction from Reed-Solomon codes. 

Given a polynomial $f$ of degree $<k/\ell$ evaluated at distinct points $1,\gamma^\ell, \gamma^{2\ell},\dotsc, \gamma^{N\ell}$, we can think of it as a degree $<k$ polynomial $g(X)=f(X^\ell)$. For $\zeta$ of order $\ell$, we have that $g(\zeta^i X)=g(X)$ for every $i$. In particular, the associated low-order folded Reed-Solomon codeword encoding $g(X)$ is simply
\begin{equation}
\label{eq:rs-red}
\begin{bmatrix}
f(1) & f(\gamma^\ell) & \hdots & f(\gamma^{N\ell})\\
f(1) & f(\gamma^\ell) & \hdots & f(\gamma^{N\ell})\\
\vdots & \vdots & \ddots & \vdots \\
f(1) & f(\gamma^\ell) & \hdots & f(\gamma^{N\ell})
\end{bmatrix}. 
\end{equation}

Notice that if $f(\gamma^{i\ell})$ is correct, then the entire $i$th column is correct, so an algorithm to list-decode the low-order folded RS code from an $\eta$ fraction of errors will also list-decode the Reed-Solomon code with evaluation points $(1,\gamma^\ell,\dots,\gamma^{N\ell})$  from the same error fraction. 
\medskip

This reduction also helps to show that the precoding used to conclude Corollary~\ref{thm:lo-frs} is necessary for a polynomial list size. 
To see this, consider the behavior of the algorithm on a transmitted codeword as in Equation~\eqref{eq:rs-red}. If there is enough agreement, the algorithm will interpolate polynomials $A_i(X)$ satisfying 
\begin{align}
0&=A_0 + A_1(X)g(X) + A_2(X) g(\zeta X) + \dotsb + A_s(X)g(\zeta^{s-1} X)\\
    &=A_0(X) + g(X)\sum_{i=1}^s A_i(X).\label{eq:bad}
\end{align}

If $\sum_{i>0} A_i(X)\neq 0$, then $g(X)$, and thus $f(X)$, can be recovered {\em uniquely} as $A_0(X)/\sum_{i>0} A_i(X)$; however, this will not be possible in general outside of the unique decoding radius. If $\sum_{i>0} A_i(X)$ is $0$, then $A_0(X)=0$ as well and {\em any} function which is a polynomial of $X^\ell$ satisfies Equation~\eqref{eq:bad}, and in particular the output list must have size at least $q^{k/\ell}$. Recall that $\ell$ is a constant in our application. 

This implies that without precoding, the dimension of the list output by Proposition~\ref{lem:decoding} over $\F_q$ must be $\Omega(k/\ell)$. Note that for the value $a=\theta(k/\ell)$ found in Section~\ref{sec:irred}, the list size before precoding would be $O(ks/\ell)$. 

\section{Conclusion and open questions}
\label{sec:conc}

We have given an explicit construction of list-decodable rank-metric
and subspace codes, which were obtained by restricting known codes to
carefully chosen subcodes. However, our results give no insight into
whether the Gabidulin and KK codes can be themselves list-decoded beyond half the
distance. We close with the following natural open problems.
\begin{itemize}
\itemsep=0ex
\item[-] Is it combinatorially feasible to list-decode Gabidulin codes {\em themselves} beyond half the distance? 
We note that it was recently shown that there is no analog of the classical Hamming-metric Johnson bound in the world of rank-metric codes always guaranteeing list-decodability beyond half the minimum distance~\cite{wachter-zeh}. Therefore, a proof of list-decodability past the unique decoding radius (say for the Gabidulin code) must account for the code structure beyond just the minimum distance.

\item[-] Assuming it is combinatorially feasible, can we give an efficient algorithm to list-decode Gabidulin codes without using subcodes or special evaluation points? 

\item[-] Currently, for rate $R$ codes, we do not know where in the range $(1-\sqrt{R},1-R)$ the list-decoding radius of Reed-Solomon codes lies, and where in the range $[(1-R)/2,1-\sqrt{R}]$ the list-decoding radius of Gabildulin codes lies. Is there a relationship between these questions?

\item[-] Can one construct better subspace-evasive sets to give an {\em explicit} code that is list-decodable from a fraction $1-R-\epsilon$ of errors with $\mathrm{poly}(1/\epsilon)$ list-size? We only known a list-size upper bound that is exponential in $1/\epsilon$ for current explicit constructions, whereas a list-size of $O(1/\epsilon)$ can be obtained with a Monte Carlo construction~\cite{GW,GX12,GX}. This question is open for errors in the usual Hamming metric also.
\end{itemize}

\section*{Acknowledgment}
We thank Antonia Wachter-Zeh for bringing to our attention the lack of a Johnson-type bound for list decoding rank-metric codes~\cite{wachter-zeh}.


\appendix

\section{Explicit list-decodable subspace codes}
\label{sec:KK}

\subsection{The operator channel and subspace codes}

For a vector space $W$, let $\mathcal{P}(W)$ denote the set of all subspaces of $W$, and $\mathcal{P}_n(W)$ the set of all $n$-dimensional subspaces of $W$. 

We recall the definition of the operator channel from~\cite{KK08}. 

\begin{defn} An {\it operator channel} $C$ associated with the {\it ambient space} $W$ is a channel with input and output alphabet $\mathcal{P}(W)$. The channel input $V$ and output $U$ are related by 
\[U=\mathcal{H}_k(V) + E,\]
where $k=\dim(U\cap V)$, $E$ is an error subspace (wlog $E$ may be taken such that $E\cap V=\{0\}$), and $\mathcal{H}_k(V)$ is an operator returning an arbitrary $k$-dimensional subspace of $V$. 

In transforming $V$ to $U$, we say that operator channel commits $r=\dim(V)-k$ {\it deletions} and $t=\dim (E)$ {\it insertions}. 
\end{defn}

A {\it subspace code} $C$ is a subset of $\mathcal{P}_n(\F_q^t)$ for some $n$. We define the {\it rate} of a subspace code to be 
\[R(C)=\frac{\log_q\abs{C}}{nt}.\]

\subsection{The K\"otter-Kschischang (KK) code}

Our constructions will be subcodes of the KK code (as introduced in~\cite{KK08}), 
which we now define. 

For $n$ dividing $t$, let $\F_{h^t}$ extend $\F_h$, and let $\alpha_1,\dotsc, \alpha_n\in\F_{h^t}$ generate the subfield $\F_{h^n}:=\F_q$. 

Set $m=t/n$. 
Then the $(n,k,t)$ \textbf{KK code} encodes an $\F_h$-linearized polynomial over $\F_{q^m}=\F_{h^t}$ of $q$-degree $<k$ by 
\[f(X)\mapsto\spam\{(\alpha_i, f(\alpha_i)\}_{i=1}^n.\] 

The encoding of $f$ is an $n$-dimensional vector space in the ambient space of dimension $n + t$ over $\F_h$. 

When $k<n$, this code has distance $2(n-k+1)$ and rate 
\[\frac{\log_h q^{mk}}{n(n+t)} = \frac{k}{n} \left(\frac{1}{1+n/t}\right)\approx \frac{k}{n} \quad \text{(when $n\ll t$)}.\]

If the channel commits $\leq \mu$ deletions and $\leq \rho$ insertions, where $s\mu + \rho < s(n-k+1)$, Guruswami and Xing~\cite{GX} give a list-decoding algorithm which outputs a $(s-1, m,k)$-periodic subspace in $\F_{q}^{mk}$ containing all candidate messages. 

\subsection{List-decodable subcodes}

By restricting the coefficients of the message polynomial $f$ to come from distinct $H_1,\dotsc, H_k$ from the $\bigl(s, 2(m-1)s/\epsilon,t\bigr)$-subspace design of Theorem~\ref{thm-design}, and setting $m\approx s/\epsilon$, we can prune the list down to a $\F_h$-subspace of dimension $O(s^2/\epsilon^2)$. 

Notice that the $H_i$'s are $\F_h$-linear subspaces, so the restricted subcode is linear. In summary, we have: 

\begin{thm}\label{thm:KK} For every $\epsilon>0$ and integer $s>0$, there exists an explicit linear subcode of the $\bigl(n,k,sn/\epsilon\bigr)$ KK code of rate $(1-\epsilon)k/n$ which is list-decodable from $\rho$ insertions and $\mu$ deletions, provided $\rho + s\mu < s(n-k+1)$. 

Moreover, the output list is contained in an $\F_h$-subspace of dimension $O(s^2/\epsilon^2)$. 
\end{thm}

\end{document}